\documentclass[11pt,a4paper]{article}
\usepackage{vmargin}

\usepackage{amsfonts,amsmath,amssymb}
\usepackage{amsthm}
\usepackage{graphicx}
\usepackage{color}

\newtheorem{theorem}{Theorem}

\newtheorem{corollary}{Corollary}

\newtheorem{observation}[theorem]{Observation}
\newtheorem{remark}{Remark}
\newtheorem{definition}{Definition}

\usepackage{pifont}
\usepackage{amsmath}
\usepackage{amsfonts,amsmath,amssymb}
\usepackage{graphicx}
\usepackage{amsbsy}
\usepackage{amsmath}
\usepackage{color}
\usepackage{todonotes}
\usepackage{algorithm}
\usepackage{algpseudocode}
\usepackage{comment}
\usepackage{mathtools}
\usepackage{complexity}
\usepackage{subcaption}
\usepackage{hyperref}
\usepackage{enumitem}

\title{Schelling Segregation with Strategic Agents\\{\small (full version)}}
\author{Ankit Chauhan\thanks{Algorithm Engineering Group, Hasso Plattner Institute Potsdam, Germany, \texttt{firstname.lastname@hpi.de}} \and Pascal Lenzner\footnotemark[1] \and Louise Molitor \footnotemark[1]}

\date{~}

\newcommand*{\cost}{\textit{cost}}
\newcommand*{\dist}{d}
\newcommand*{\fav}{\textit{fav}}

\newcommand*{\RoH}{\textrm{RoH}}
\newcommand*{\RoD}{\textrm{RoD}}

\newcommand*{\opt}{\textrm{opt}}
\newcommand{\pg}{p_G}

\begin{document}

\maketitle

\begin{abstract}
 \noindent Schelling's segregation model is a landmark model in socio\-logy. It shows the counter-intuitive phenomenon that residential segregation between individuals of different groups can emerge even when all involved individuals are tolerant. Although the model is widely studied, no pure game-theoretic version where rational agents strategically choose their location exists. We close this gap by introducing and analyzing generalized game-theoretic models of Schelling segregation, where the agents can also have individual location preferences. 
 
 For our models we investigate the convergence behavior and the efficiency of their equilibria. In particular, we prove guaranteed convergence to an equilibrium in the version which is closest to Schelling's original model. Moreover, we provide tight bounds on the Price of Anarchy. 
 \end{abstract}
 
 \section{Introduction}
 Segregation is a well-known sociological phenomenon which is intensely monitored and investigated by sociologists and economists. It essentially means that a community of people which is mixed along e.g. ethical, racial, linguistic or religious dimensions tends to segregate over time such that almost homogeneous sub-communities emerge. The most famous example of this phenomenon is residential segregation along racial lines in many urban areas in the US.\footnote{See the \href{https://demographics.coopercenter.org/Racial-Dot-Map/}{racial dot map}~\cite{C13} for an impressive visualization.}           
 
 To explain the emergence of residential segregation Schelling~\cite{Schelling71} proposed in a seminal paper a very simple and elegant agent-based model.  
 In Schelling's model two types of agents, say type $A$ and type $B$ agents, are placed on a line or a grid which models some residential area. Each agent is aware of its neighboring agents and is content with her current residential position if at least a $\tau$ fraction of agents in her neighborhood is of the same type, for some $0\leq \tau \leq 1$. If this condition is not met, then the agent becomes discontent with her current position and exchanges positions with a randomly chosen discontent agent of the other type or the agent jumps to a randomly chosen empty spot.\footnote{A playful \href{http://ncase.me/polygons/}{interactive demonstration} can be found in~\cite{H16}.} Schelling showed with simple experiments using coins, graph paper and random numbers that even with $\tau \leq \frac{1}{2}$, i.e., with tolerant agents, the society of agents will eventually segregate into almost homogeneous communities. This surprising observation caught the attention of many economists, physicists, demographers and computer scientists who studied related random models and verified experimentally that tolerant local neighborhood preferences can nonetheless induce global segregation in social and residential networks, see e.g.~\cite{RAL93,BW07,Wil08,HPP11}. 
 
 To the best of our knowledge, all agent-based models of segregation are essentially random processes where discontent agents choose their new location at random. In this paper we depart from this assumption by introducing and analyzing a game-theoretic version of Schelling's model where agents strategically choose their location. Empirically, our model yields outcomes which are very similar to Schelling's original model - see Fig.~\ref{fig:sample} for an example.
 \begin{figure}[h!]
 \centering
 \includegraphics[width=\textwidth]{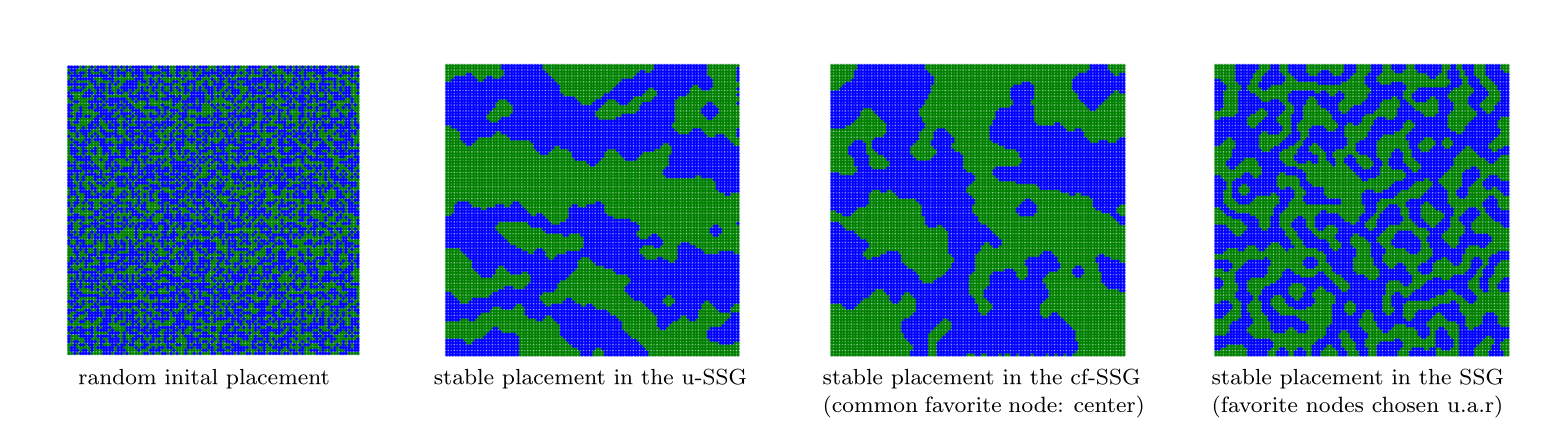}\vspace*{-0.3cm}
 \caption{Sample equilibria of our model showing significant segregation. Here $|A|=|B|=5000$, nodes of type $A$ are blue, type $B$ nodes are green and $\tau = \frac{1}{2}$.  \vspace*{-0.5cm}}
 \label{fig:sample}
 \end{figure}

 \noindent Moreover, our model generalizes Schelling's model since we allow agents to have preferences over the available locations. Hence, we introduce and explore the influence of such individual location preferences.
 
 \subsection{Related Work}
 There is a huge body of work on Schelling's model and variations thereof, see e.g. \cite{Sch69,S06,Cla86,Whi86,BW07,IEE09,Pan07,You98,Vin06}. Most related work is purely empirical and provides simulation results. We focus here on the surprisingly small amount of related work, which rigorously proves properties of (variants of) Schelling's model.
 
Young~\cite{You98} was the first to rigorously analyze a variant of the one-dimensional segregation model by using techniques from evolutionary game theory. He considered the specific dynamics where a pair of agents is chosen at random and they swap places with a suitably chosen probability. Then he analyzes the induced Markov chain and proves that under certain conditions total segregation will be with high probability a stochastically stable state. Later Zhang~\cite{Zha04,Zha04b} proved similar results in 2-dimensional models.    

The first rigorous analysis of the original Schelling model was achieved by Brandt et al.~\cite{BIK12} for the case where agents with tolerance parameter $\tau = \frac{1}{2}$ are located on a ring and agents can only swap positions. They prove that the process converges with high probability to a state where the average size of monochromatic neighborhoods is polynomial in $w$, where $w$ is the window-size for determining the neighborhood. Interestingly, Barmpalias et al.~\cite{BEL14} have proven a drastically different behavior for $0.3531 < \tau < \frac{1}{2}$ where the size of monochromatic neighborhoods is exponential in $w$. Later, Barmpalias et al.~\cite{BEL16} analyzed a 2-dimensional variant where both agent types have different tolerance parameters and agents may change their type if they are discontent. Finally, Immorlica et al. \cite{BIK17} considered the random Schelling dynamics on a 2-dimensional toroidal grid with $\tau=\frac{1}{2}-\epsilon$, for some small $\epsilon>0$. Their main result is a proof that the  average size of monochromatic neighborhoods is exponential in $w$.

Not much work has been done on the game theory side. To the best of our knowledge only the model by Zhang~\cite{Zha04b} is game-theoretic and closely related to Schelling's model.
In this game agents are placed on a toroidal grid and are endowed with a noisy single peaked utility function which depends on the ratio of the numbers of the two agent types in any local neighborhood. The highest utility is attained in perfectly balanced neighborhoods and agents slightly prefer being in the majority over being in the minority. In contrast to our model, Zhang's model~\cite{Zha04b} assumes transferable utilities and it can happen that after a randomly chosen swap one or both agents are worse off. Moreover Zhang's model does not incorporate the threshold behavior at $\tau$. 
However, despite the different model Zhang~\cite{Zha04b} uses a similar potential function as we do in this paper. 

We note that hedonic games~\cite{DG80,BJ02} are also remotely related to Schelling's model, but there the utility of an agent only depends on her chosen coalition. In Schelling's model the neighborhood of an agent could be considered as her coalition, but then not all agents in a coalition derive the same utility from it.

\subsection{Model and Notation}
We consider a network $G = (V,E)$, where $V$ is the set of nodes and $E$ is the set of edges, which is connected, unweighted and undirected. If in $G$ every node has the same degree, i.e., the same number of incident edges, then we call $G$ \emph{regular}.
The distance $\dist_G(u,v)$ between two nodes $u,v \in V$ in network $G$ is the number of edges on a shortest path between $u$ and $v$. The diameter of $G$ is the length of the longest shortest path between any pair of nodes and is denoted by $D(G)$. For a given node $u \in V$ let $N_w(u)$ be the set of nodes $v \neq u$ which are in distance at most $w$ from node $u$. We call $N_w(u)$ the $w$-neighborhood of $u$ and $w\geq 1$ is the window size. We will omit $w$ whenever a statement holds for all $w\geq 1$.

Agents of two different types are located on the nodes of network $G$. There are two disjoint sets of agents $A$ and $B$, with $|A|, |B| \geq 2$ and we say that all agents $a \in A$ are of type $A$ and agents  $b\in B$ are of type $B$. In each state of our game, there is an injective mapping $\pg:\{A\cup B\}\rightarrow V$ between agents and nodes which we call a \emph{placement}. In any placement $\pg$ a node of $G$ can be occupied by exactly one agent either from $A$ or from $B$ or the node can be empty. Let $\pg$ be any placement and let $x,y$, with $x\in A$ and $y \in B$, be agents which are neighbors under placement $\pg$. In this case, we call $x,y$ a \emph{colored pair}.

For any agent $x \in A \cup B$, $\pg(x) = u$, we define $N^+_w(\pg(x)) \subseteq N_w(\pg(x))$, as the set of other nodes $v$ in the $w$-neighborhood of node $u$, with $u\neq v$, which are occupied by the same type of agents as agent $x$ and $N^-_w(\pg(u)) \subseteq N_w(\pg(u))$ is the corresponding set of other nodes which are occupied by agents of the other type. Note that $\pg(x) \notin N^+_w(\pg(x))$. If $|N^+_w(\pg(x))| + |N^-_w(\pg(x))| = 0$, then agent $x$ has no neighboring agents and we say that agent $x$ is \emph{isolated}. 

Let $\tau \in [0,1]$ be the \emph{tolerance parameter}. Similar to Schelling's model we say that an un-isolated agent $x$ is \emph{happy} or \emph{content} with placement $\pg$ if at least a $\tau$-fraction of the agents which occupy the nodes in her $w$-neighborhood under $\pg$ are of the same type as her. I.e., an un-isolated agent $x$ is happy if $\tfrac{|N^+_w(\pg(x))|}{|N^+_w(\pg(x))|+|N^-_w(\pg(x))|} \geq \tau,$ otherwise $x$ is \emph{unhappy} or \emph{discontent} with placement $\pg$. Moreover, we will assume that isolated agents are always unhappy. We call the ratio $\tfrac{|N^+_w(\pg(x))|}{|N^+_w(\pg(x))|+|N^-_w(\pg(x))|}$ the \emph{local happiness ratio} of agent $x$.
Besides having preferences about the neighborhood structure, every agent may have a \emph{favorite node} $\fav_x \in V$ in the network $G$.

The cost function of our agents is based on two main assumptions:
\begin{quote}
\begin{itemize}[nosep]
 \item[(1)] An agent's high priority goal is to find a location where she is happy. 
 \item[(2)] An agent's low priority goal is to find a location which is as close as possible to her favorite location. 
\end{itemize}
\end{quote}
Thus, a happy agent $x$ strives for locations where she is happy, but as close as possible to $\fav_x$. If an agent $x$ is unhappy, she will try to improve her local happiness ratio. If this is not possible then she will select a location which has maximum possible local happiness ratio and which is closest to $\fav_x$.  

We incorporate these assumptions as follows in our cost function:
The cost of an un-isolated agent $x$ with placement $\pg$ in network $G$ is the vector 
\[\cost_x(\pg) = \left(\max\left(0,\tau - \frac{|N^+_w(\pg(x))|}{|N^+_w(\pg(x))|+|N^-_w(\pg(x))|}\right),\dist_G(\fav_x,\pg(x)) + 1\right).\]
For an isolated agent $x$ the cost vector is $\cost_x(\pg) = \left(\tau,\ \dist_G(\fav_x,\pg(x)) + 1\right)$.

Thus, an agent $x$ is happy with placement $\pg$, if and only if $\cost_x(\pg) = (0,\cdot)$. Note that we use $d_G()+1$ instead of $d_G()$ as second component of the cost vector for technical reasons. This has no influence on the behavior of the agents. 

We choose the lexicographic order $\leq_{lex}$ \footnote{$
 (\alpha,\beta) <_{lex} (\gamma,\delta) \textnormal{, if }\alpha < \gamma \textnormal{ or } \alpha = \gamma \textnormal{ and }\beta < \delta$. $
 (\alpha,\beta) =_{lex} (\gamma,\delta) \textnormal{, if }\alpha = \gamma\textnormal{ and }\beta = \delta$. $
 (\alpha,\beta) >_{lex} (\gamma,\delta) \textnormal{, if }\alpha > \gamma \textnormal{ or } \alpha = \gamma \textnormal{ and }\beta > \delta$.} for comparing cost vectors.
Agents want to minimize their cost vector lexicographically, i.e., it is more important for an agent to be happy than to be close to her favorite node.

The \emph{social cost} $\cost(\pg)$ of a placement $\pg$ in network $G$ is the vector consisting of the number of unhappy players and the sum of all distance terms: 
\begin{flalign*}
 &\cost(\pg) =\\
 &\hspace*{1.5 em}\left(|\left\{x \in A \cup B \mid \cost_x(\pg) = (\alpha,\cdot), \alpha \neq 0\right\}|, \sum_{x \in A\cup B} (\dist_G(\fav_x,\pg(x)) + 1)\right). 
\end{flalign*}
The \emph{strategy space} of an agent is the set of all nodes of $G$. A strategy vector is \emph{feasible} if all of its entries are pairwise disjoint. Clearly, there is a bijection between feasible strategy vectors and placements $\pg$ and we will use them interchangeably.
For the possible strategy changes of an agent there are two versions, which yield the \emph{Swap Schelling Game} and the \emph{Jump Schelling Game}.

\textbf{The Swap Schelling Game:} In the \emph{Swap Schelling Game (SSG)} only pairs of agents can jointly change their strategies by swapping locations. Two agents $x$ and $y$ agree to swap their nodes if both agents strictly decrease their cost by swapping. A placement $\pg$ is \emph{stable} if no pair of agents can both improve their cost via swapping. Hence, stable placements correspond to 2-coalitional pure Nash equilibria.  
Since locations can only be swapped, we will assume throughout the paper that there are no empty nodes in $G$, that is, $\pg$ is also surjective.

\textbf{The Jump Schelling Game:}
In the \emph{Jump Schelling game (JSG)} an agent can change her strategy to any currently empty node, which constitutes a ``jump'' to that node. An agent will jump to another empty node, if this strictly decreases her cost. Here a \emph{stable} placement $\pg$ corresponds to a pure Nash equilibrium.

\textbf{Different Variants:} Besides assuming that every agents has some individual favorite position, we will consider two additional variants of the SSG and the JSG, depending on the favorite nodes of the agents. If the agents do not have a favorite node, then we call these versions \emph{uniform} (\emph{u-SSG} or \emph{u-JSG}) and we simply ignore the second entry in the cost vector. Note that the uniform versions are very close to Schelling's original model. If all agents have the same favorite node, then we call these games \emph{common favorites} (\emph{cf-SSG} or \emph{cf-JSG}). Observe that this variant is especially interesting since it models the case where some particular location is intrinsically more attractive than others to all agents, e.g. it could be the most popular location in a city.

\textbf{Dynamic Properties:} We will use \emph{ordinal potential functions}. Such a function $\Phi$ maps placements to real numbers with the property that if $\pg'$ is the placement which results from an improving move by a (pair of) agent(s) under placement $\pg$ then $\Phi(\pg') < \Phi(\pg)$. 
If an ordinal potential function for some special case of the game exists, then this implies that this special case has the \emph{finite improvement property (FIP)}, which states that any sequence of improving moves must be finite. Having the FIP is equivalent to the game being a potential game~\cite{MS96}. Such games have many attractive properties like guaranteed existence of pure equilibria and often a fast convergence to such a stable state. Moreover, a potential function is useful for analyzing the quality of equilibria. 
In contrast, if an infinite sequence of improving moves, usually called \emph{improving response cycle (IRC)}, exists then there cannot exist an ordinal potential function.

\subsection{Our Contribution}
We introduce the first agent-based model for Schelling segregation where the agents strategically choose their locations. For this, we consider a generalization of Schelling's model where agents besides having preferences over their local neighborhood structure also have preferences of the possible locations. This introduces the important aspect of individual location differentiation which has a significant influence on residential decisions in real life. 

Our main contribution is a thorough investigation of the convergence properties of many variants of our model. See Table~\ref{overviewtable} for details.

\begin{table}[h!]
\centering 
\resizebox{\textwidth}{!}{\begin{tabular}{l|c|c|c|c|c}
& u-SSG & cf-SSG & SSG & u-JSG & cf-JSG \& JSG \\ \hline
$\tau < \frac{1}{2}$ & $\checkmark$ (T.\ref{theorem_potential_game}) & reg. (T.\ref{thm:potentialgameSSGsingle})& ring, $w=1$ (T.\ref{thm:SSG-ring}) & ring, $w=1$ (T.\ref{thm:JSG-ring}) & IRC: $\frac{1}{3} < \tau <\frac{1}{2}$ (T.\ref{thm:IRC_cfJSG}) \\ \hline
$\tau = \frac{1}{2}$ & $\checkmark$ (T.\ref{theorem_potential_game})& reg. (T.\ref{thm:potentialgameSSGsingle})& reg. (T.\ref{thm:potentialgameSSGmorethanhalf}) & ring, $w=1$ (T.\ref{thm:JSG-ring}) & IRC (T.\ref{thm:IRC_cfJSG}) \\ \hline
$\tau > \frac{1}{2} $ & reg. (T.\ref{thm:potentialgameSSG})& reg. (T.\ref{thm:potentialgameSSGsingle})& reg. (T.\ref{thm:potentialgameSSGmorethanhalf})& ring, $w=1$ (T.\ref{thm:JSG-ring}) & IRC: $\frac{1}{2} < \tau \leq  \frac{2}{3}$(T.\ref{thm:IRC_cfJSG})\vspace*{0.2cm}
\end{tabular}}
\caption{Convergence results. ``$\checkmark$'': potential game for any $w$ and $G$. ``reg.'': potential game for any $w$ and regular networks $G$. ``ring'': potential game on a ring. IRC: an improving response cycle exists, i.e., not a potential game. }
\label{overviewtable}
\end{table}

In particular, we prove guaranteed convergence to an equilibrium for u-SSG, which essentially is Schelling's model, if tolerant agents are restricted to location swaps or if the underlying network is regular. In contrast, previous work~\cite{BIK12,BEL14,BEL16,BIK17} has established, that the process converges with high probability. Moreover, also the (cf)-SSG behaves nicely on regular networks. In contrast to this, we show that location preferences have a severe impact in the (cf-)JSG, since improving response cycles exist, which imply that there cannot exist a potential function. 

Furthermore, we investigate basic properties of stable placements and their efficiency in the (u-)SSG. In particular, we prove tight bounds on a variant of the Price of Anarchy for the (u-)SSG.    
\section{Dynamic Properties}\label{sec:dynamics}
We analyze the convergence behavior of the Schelling game. Our main goal is to investigate under which conditions an ordinal potential function~$\Phi$ exists.

\subsection{Dynamic Properties of the Swap Schelling Game} 
We prove for various special cases of the SSG that they are actually potential games. For this we analyze the change in the potential function value for a suitably chosen potential function $\Phi$ for an arbitrary location swap of two agents $x$ and $y$. Such a swap changes the current placement $\pg$ only in the locations of agents $x$ and $y$ and yields a new placement~$\pg'$.

\begin{theorem}
If $\tau \leq \frac{1}{2}$ then the u-SSG is a potential game for any $w \geq 1$. \label{theorem_potential_game}
\end{theorem}
 
\label{potentialgame}
\begin{proof}
We prove the statement by showing that
$$\Phi(\pg) = \frac{1}{2}\sum_{x \in A\cup B}|N^-(\pg(x))|$$ is an ordinal potential function. Note that $\Phi$ is the number of colored pairs. 
First of all, notice that a swap between two agents $x \in A$ and $y \in B$ will only executed when both agents are unhappy and of different types, since a swap between agents of the same type cannot be an improvement for at least one of the involved agents. Furthermore a happy agent has no possibility to improve, so there is no incentive to. An agent will decrease her cost if and only if she is unhappy and reduces the ratio of neighbors with different type by swapping.  It holds that
\[\frac{|N^+(\pg(x))|}{|N^+(\pg(x))|+|N^-(\pg(x))| } < \tau\ \text{and } \frac{|N^+(\pg(y))|}{|N^+(\pg(y))|+|N^-(\pg(y))| } < \tau. \]
Hence, 
\[\frac{|N^-(\pg(x))|}{|N^+(\pg(x))|+|N^-(\pg(x))| } > 1-\tau\ \text{and } \frac{|N^-(\pg(y))|}{|N^+(\pg(y))|+|N^-(\pg(y))| } > 1-\tau. \]
Since $\tau \leq \frac{1}{2}$ it follows that 
\begin{align*}|N^+(\pg(x))| &< \tau \cdot (|N^+(\pg(x))|+|N^-(\pg(x))| ) \\ &\leq\ (1-\tau)\cdot (|N^+(\pg(x))|+|N^-(\pg(x))| ) \\&<  |N^-(\pg(x))|
\end{align*}
and analogously we get for agent $y$ that $|N^+(\pg(y))| < |N^-(\pg(y))|$. Thus,
\begin{alignat*}{2}
&\ |N^+(\pg(x))| + |N^+(\pg(y))| < |N^-(\pg(x))| + |N^-(\pg(y))| \\
= &\ |N^-(\pg'(y))| + |N^-(\pg'(x))| < |N^-(\pg(x))| + |N^-(\pg(y))|.
\end{alignat*} 
This implies that for the change in the potential function value,

 $\Phi(\pg)-\Phi(\pg')=(|N^-(\pg(x))| + |N^-(\pg(y))|)-(|N^-(\pg'(y))| + |N^-(\pg'(x))|)>0.$ 
\end{proof}

\begin{remark}
The function $\Phi(\pg) = \frac{1}{2}\sum_{x \in A\cup B}|N^-(\pg(x))|$ is not a potential function for the (cf-)SSG. See Fig.~\ref{fig:potential_function_position} below.
\end{remark}
\begin{figure}[h!]
\centering
\includegraphics[scale=0.4]{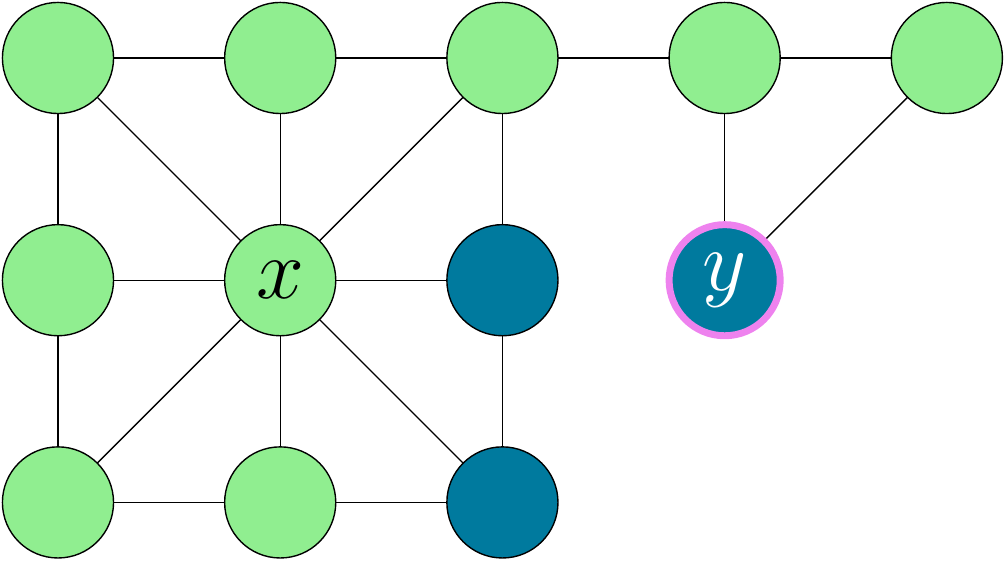}
\caption{Consider the two agents $x$ and $y$ for $\tau = 0.3$. Let the current node of $y$ be the favorite node $\fav_x=\fav_y$ of agent $x$ and agent $y$. A swap is an improving move for both agents, since agent $y$ reduces her cost from $(0.3,1)$ to $(0.05, 4)$ and $x$ reduces her cost from $(0,4)$ to $(0,1)$. However, the potential function changes by $2 + 2 - 6 = -2$ and therefore increases.}
\label{fig:potential_function_position}
\end{figure}

\pagebreak

\begin{theorem}\label{thm:potentialgameSSG} For any $w\geq 1$ the u-SSG on regular networks is a potential game. 
\end{theorem}
\begin{proof}
We prove the statement by showing that $$\Phi(\pg) = \frac{1}{2}\sum_{x \in A\cup B}|N^-(\pg(x))|$$ is an ordinal potential function. 
Analogously to the proof of Theorem~\ref{theorem_potential_game}, there is no incentive for a happy agent to swap or for an unhappy agent $x\in A$ to swap with another agent $y \in A$. So an unhappy agent $x \in A$ will only swap with another agent $y \in B$ of different type.

Since we consider the u-SSG on regular networks, we have \[\forall x \in A \cup B\colon |N(\pg(x))| = |N^+(\pg(x))| + |N^-(\pg(x))| = k. \]
For any improving swap between $x\in A$ and $y\in B$ it follows that
$$\frac{|N^+(\pg(x))|}{k} < \frac{|N^+(\pg'(x))|}{k} \text{ and } \frac{|N^+(\pg(y))|}{k} < \frac{|N^+(\pg'(y))|}{k}.$$
We have to distinguish between two cases: 
\begin{itemize}
\item Assume $x$ and $y$ are in different neighborhoods, so $x \notin N_w(\pg(y))$ and $y \notin N_w(\pg(x))$. Observe that all agents who were in $N^+(\pg(x))$ or $N^+(\pg(y))$ before the swap are after the swap in $N^-(\pg'(y))$ or $N^-(\pg'(x))$, respectively. The same holds the other way around. 
Hence, 
$$|N^-(\pg'(y))| = |N^+(\pg(x))| < |N^+(\pg'(x))|=|N^-(\pg(y))|$$ 
and 
$$|N^-(\pg'(x))| = |N^+(\pg(y))| < |N^+(\pg'(y))| = |N^-(\pg(x))|.$$
\item Assume $x$ and $y$ are in the same neighborhood, so $x \in N_w(\pg(y))$ and $y \in N_w(\pg(x))$. All agents who were in $N^+(\pg(x))$ or $N^+(\pg(y))$ before the swap are after the swap in $N^-(\pg'(y))$ or $N^-(\pg'(x))$, respectively. Except for $x$ and $y$ all agents who were in $N^-(\pg(y))$ or $N^-(\pg(x))$ are after the swap in $N^+(\pg'(x))$ or $N^+(\pg'(y))$, respectively. The two agents $x$ and $y$ who are involved in the swap are after the swap in $N^-(\pg'(y))$ and $N^-(\pg'(x))$, respectively.
Hence $$|N^-(\pg'(y))| - 1 = |N^+(\pg(x))| < |N^+(\pg'(x))|=|N^-(\pg(y))| - 1$$ and 
$$|N^-(\pg'(x))| - 1 = |N^+(\pg(y))| < |N^+(\pg'(y))| = |N^-(\pg(x))| - 1.$$

\end{itemize}
Since a swap between two agents $x$ and $y$ just affects colored pairs where $x$ or $y$ are involved, we have that $$\Phi(\pg) - \Phi(\pg') =|N^-(\pg(x))| + |N^-(\pg(y))| - (|N^-(\pg'(x))| + |N^-(\pg'(y))|) > 0. \qedhere$$ 
\end{proof}

\begin{remark} The function $\Phi(\pg) = \frac{1}{2}\sum_{x \in A\cup B}|N^-(\pg(x))|$ is not a potential function for the u-SSG on non-regular networks. See Fig.~\ref{fig:potential_function}.
\end{remark}
\begin{figure}[h!]
\centering
\includegraphics[scale=0.4]{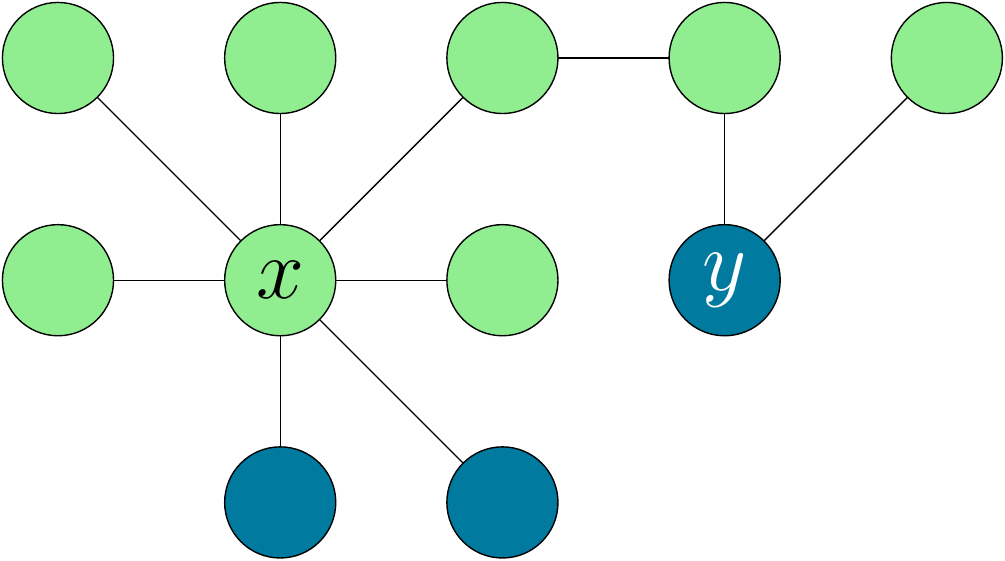}
\caption{Consider the two agents $x$ and $y$ for $\tau = \frac{6}{7}$. A swap is an improving move for both agents, since agent $x$ is happy afterwards and reduces her cost from $(\frac{1}{7},1)$ to $(0,1)$ and $y$ reduces her cost from $(\frac{6}{7},1)$ to $(\frac{4}{7},1)$. However, the potential function changes by $2 + 2 - 5 = -1$ and therefore increases.}
\label{fig:potential_function}
\end{figure}

\begin{theorem}\label{thm:potentialgameSSGsingle} For any $w \geq 1$ the cf-SSG is a potential game on regular networks. 
\end{theorem}

\begin{proof} 
Analogously to the proof of Theorem \ref{thm:potentialgameSSG} we prove the statement by showing that $$\Phi = \frac{1}{2}\sum_{x \in A\cup B}|N^-(\pg(x))|$$ is an ordinal potential function. 

Again, note that there is no incentive for an agent $x \in A$ to swap with another agent $y \in A$ who has the same type, since this swap is a setback for at least one of the involved agents, as she either becomes unhappy or increases her distance cost. Without loss of generality, we can assume $x \in A$ and $y \in B$.

If $x$ and $y$ are both happy before the swap, they would also not swap. Either at least one of them becomes unhappy or, since both have the same favorite node $\fav$, the decrease in the distance cost of one agent is equal to the increase of the distance cost of the other agent. Thus to swap is not an improving move for at least one of the involved agents.

Hence, we just have to consider the case where at least one agent is unhappy.
First, we assume without loss of generality that $y$ is unhappy before the swap and $x$ is happy.
Again, since we have regular networks, we have \[\forall x' \in A \cup B\colon |N(\pg(x'))| = |N^+(\pg(x'))| + |N^-(\pg(x'))| = k. \]
Agent $x$ is just willing to swap if she stays happy and gets closer to her favorite node $\fav$. For this reason a possible swap will increase agent $y$'s distance cost. Therefore, since a swap has to be an improvement for all involved agents, we have $$\frac{|N^+(\pg(y))|}{k} <\frac{|N^+(\pg'(y))|}{k}.$$
Analogous the proof of Theorem \ref{thm:potentialgameSSG} we have to distinguish between two subcases:
\begin{itemize}
\item  Assume $x$ and $y$ are in different neighborhoods, so $x \notin N_w(\pg(y))$ and $y \notin N_w(\pg(x))$. It holds that \[|N^-(\pg'(x))| = |N^+(\pg(y))| <|N^+(\pg'(y))| = |N^-(\pg(x))|\]
and \[k - |N^-(\pg(y))| =|N^+(\pg(y))| <|N^+(\pg'(y))| = k - |N^-(\pg'(y))|.\]
\item   Assume $x$ and $y$ are in the same neighborhood, so $x \in N_w(\pg(y))$ and $y \in N_w(\pg(x))$. It holds that \[|N^-(\pg'(x))| - 1 = |N^+(\pg(y))| < |N^+(\pg'(y))| = |N^-(\pg(x))| - 1\] and \[k - |N^-(\pg(y))| =|N^+(\pg(y))| <|N^+(\pg'(y))| = k - |N^-(\pg'(y))|.\]
\end{itemize} 
Which leads immediately to \[|N^-(\pg'(x))| < |N^-(\pg(x))| \text{ and } |N^-(\pg'(y))| < |N^-(\pg(y))|.\]
Since a swap between two agents $x$ and $y$ just affects colored pairs where agent $x$ or $y$ are involved, we have$$\Phi(\pg) - \Phi(\pg')  =  (|N^-(\pg(x))| + |N^-(\pg(y))|) - (|N^-(\pg'(x))| + |N^-(\pg'(y))|) > 0.$$
The second case we have to consider is when both agents $x$ and $y$ are unhappy before the swap. If the distance cost is equal for both agents, which means that the swap does not change the respective distances to $\fav$, then we already showed in the proof of Theorem~\ref{thm:potentialgameSSG} that a swap would decrease the potential function. Otherwise one agent, without loss of generality agent $y$, would increase her distance cost. So we must have \[\frac{|N^+(\pg(y))|}{k} <\frac{|N^+(\pg'(y))|}{k},\] which leads to the previous case. 
\end{proof}

\begin{theorem}\label{thm:potentialgameSSGmorethanhalf}
For any $w \geq 1$ and any $\tau\geq \frac{1}{2}$ the SSG is a potential game on regular networks.
\end{theorem}
\begin{proof}
We prove the theorem by showing that 
$$\Phi(\pg) = \left(\frac{1}{2} \sum_{x \in A\cup B} |N^-(\pg(x))|,\ \sum_{x \in A\cup B} \dist_G(\fav_x,\pg(x))\right)$$
is an ordinal potential function. We prove that if the agents $x$ and $y$ make an improving swap, the value of the potential function $\Phi$ decreases lexicographically.

Since we have regular networks, we have \[\forall x \in A \cup B\colon |N(\pg(x))| = |N^+(\pg(x))| + |N^-(\pg(x))| = k. \]
From the proof of Theorem \ref{thm:potentialgameSSG} we already know that whenever an agent $x$ who is involved in the swap improves her local happiness ratio, which means \[\frac{|N^+(\pg(x))|}{k} <\frac{|N^+(\pg'(x))|}{k},\] this always leads to a decrease of $\frac{1}{2} \sum_{x \in A\cup B} |N^-(\pg(x))|$. Also, from the definition of SSG we know that every agent tries first to maximize her local happiness ratio. This implies $\Phi$ decreases lexicographically.

So we just have to consider the case when two agents $x$ and $y$ swap only to decrease their distance cost.

When both agents $x,y \in A$(or $B$) swaps, the number of colored pairs stays the same, so the first entry of $\Phi$ doesn't change. Since the agents just swap when both decreases their distance cost, the sum over all distance cost decreases, which reduces the second entry of $\Phi$.

Now it can be observed that the two happy agents $x\in A$ and $y\in B$ can swap if $\tau=\frac{1}{2}$ where $N^-(\pg(x))=N^+(\pg(x))$ and $N^-(\pg(y))=N^+(\pg(y))$ and $x$ and $y$ are in different neighborhoods, so $x \notin N_w(\pg(y))$ and $y \notin N_w(\pg(x))$. In this swap, again the number of colored pairs stays the same, so the first entry of $\Phi$ doesn't change. However as a swap improves the distance cost for both of the agents thus, the sum over all distance cost decreases, which reduces the second entry of $\Phi$.  
 
If $\tau>\frac{1}{2}$ the happy agents of different types would not swap as both will become unhappy after the swap. 

Without loss of generality for $\tau\geq \frac{1}{2}$ we can assume $x\in A$ is unhappy. 

Agent $x$ swaps when $|N^+(\pg(x))|$ and $|N^-(\pg(x))|$ stays the same. Otherwise, if $x$ increases the number of agents of type $A$ on swap, so $|N^+(\pg(x))| < |N^+(\pg'(x))|$, we are in the case already discussed further above and if $|N^+(u)|$ decreases the swap is not an improvement for $u$.

Since $|N^+(\pg(x))| = |N^+(\pg'(x))|$ and $|N^-(\pg(x))| = |N^-(\pg'(x))|$, same holds for agent $y$ since the graph is regular. The number of colored pairs is unchanged. Again, both agents swap to decrease their distance to their favorite node which leads to a decrease in the sum over all distance cost, which reduces the potential function. 
\end{proof}

\begin{theorem}\label{thm:SSG-ring}
If $\tau<\frac{1}{2}$ and $w=1$ then the SSG on a ring is a potential game.
\end{theorem}
\begin{proof} We use an argument similar to the one in the proof of Theorem \ref{thm:potentialgameSSGmorethanhalf} and prove that the value of $\Phi$ decreases lexicographically with every improving swap, where $\Phi$ is defined as follows:
$$\Phi(\pg) = \left(\frac{1}{2} \sum_{x \in A\cup B} |N^-(\pg(x))|,\ \sum_{x \in A\cup B} \dist_G(\fav_x,\pg(x))\right).$$
Since $\tau < \frac{1}{2}$, an agent will become happy if she has at least one neighbor of her type. We have already shown in the proof of Theorem~\ref{thm:potentialgameSSG} that if a swap improves the local happiness ratio of an involved agent $x$ then this kind of swap always reduces the first entry of $\Phi$, thus $\Phi$ decreases lexicographically. 

Now we look at the cases when an agent reduces her distance cost by swapping. The case where at least one unhappy agent is involved in the swap is analogous to the proof of Theorem~\ref{thm:potentialgameSSGmorethanhalf}. Thus, we are left to consider the case where two happy agents swap. 

Two happy agents will swap if and only if they remain happy after the swap and if they get closer to their favorite node by swapping.  We show that after such a swap, the number of colored pairs stays the same. This implies that $\Phi$ decreases lexicographically by such a swap, since the first entry stays the same but the second entry decreases.

If two happy agents of the same type swap, then trivially nothing changes. If two happy agents $x$ and $y$ of different types swap, then both of them must have exactly one neighboring agent of their own type and of the other type, respectively and $x$ and $y$ have to be in different neighborhoods, so $x \notin N_w(\pg(y))$ and $y \notin N_w(\pg(x))$.  Thus, also in this case the number of colored pairs does not change.  
\end{proof}

\subsection{Dynamic Properties of the Jump Schelling Game}
Now we consider the JSG. Remember that in the JSG we assume that agents can only decrease their cost by jumping to empty nodes. Such a jump of an agent $x$ changes the current placement $\pg$ only in the location of agent $x$. We prove for the ring network that the u-JSG is a potential game. Furthermore we show that the cf-JSG and JSG are not potential games for different ranges of $\tau$.

\begin{theorem}\label{thm:JSG-ring}
If $w=1$ and the underlying graph is ring network then, the u-JSG is a potential game.
\end{theorem}
\begin{proof}
For any ring network $G = (V,E)$ we define the weight $w_e$ of any edge $e=(u,v)\in E$ as follows:
\[ w_e=\begin{cases} 
      1 &\text{, if $u$ and $v$ have  agents of different type}  \\
      \frac{1}{3} & \text{, if either $u$ or $v$ is empty} \\
      0 & \text{, otherwise.} 
   \end{cases}
\] 
Then we use the function
$\Phi(\pg)=\sum_{e\in E} w_e$ and prove that if any agent makes an improving jump to some other node in the ring then $\Phi$ decreases. 

An agent $x$ will jump to a new node if this improves her local happiness ratio. Let $\pg$ be the placement before the jump and let $\pg'$ be the induced placement after the jump. Thus, we have  
$$\frac{|N^{+}(\pg'(x))|}{|N^+(\pg'(x))|+|N^-(\pg'(x))|}> \frac{|N^{+}(\pg(x))|}{|N^+(\pg(x))|+|N^-(\pg(x))|}.$$ 
 Since $G$ is a ring, this implies that an unhappy agent jumps to a different node only if she increases the number of neighboring agents of her type.  Since the degree of every node in $G$ is $2$ the jump of an agent $x$ can affect the weight of at most four edges: the two edges incident to the node $\pg(x)$  and the two edges incident to the node $\pg'(x)$. Hence, the change in the weight of the involved edges equals the total change in potential function value $\Phi$. 

Without loss of generality let $x\in A$ be an unhappy agent who executes an improving jump. Depending on the neighboring nodes of $x$ we have the following cases:
\begin{itemize}
\item[(a)] Both neighboring agents of agent $x$ are of different type.
\item[(b)] Both neighboring nodes are empty.
\item[(c)] One neighboring node is empty and the other one is occupied by an agent of the other type.
\item[(d)] If one neighboring agent is of $x$'s type and the other is of the other type. (Only interesting for $\tau > \frac{1}{2}$.)
\end{itemize}
It is easy to check, than in all cases the potential function $\Phi$ decreases. See Fig.~\ref{fig:ring_network}.
\begin{figure}[h!]
\centering
\begin{subfigure}[c]{1.0\linewidth}
\centering
\includegraphics[width=0.8\linewidth]{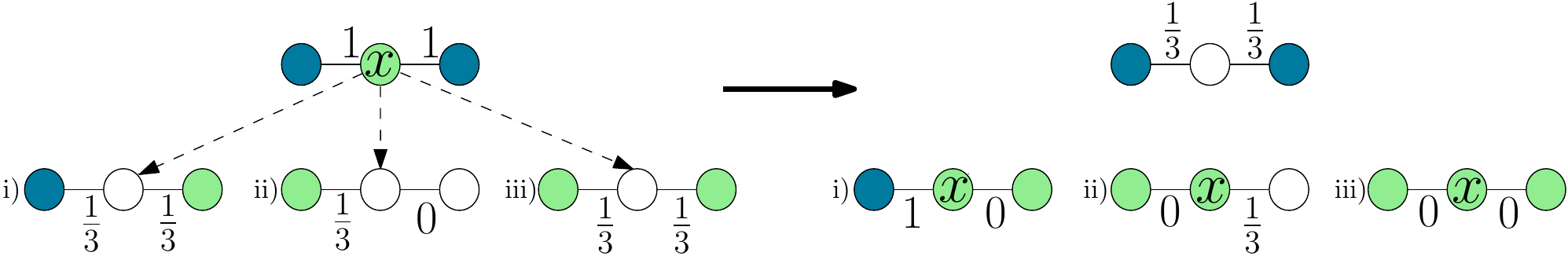}
\captionsetup{singlelinecheck=off}\subcaption{ \begin{itemize} \setlength{\itemindent}{5em}\vspace*{-0.5cm}
\item[(i)]$\Phi(\pg) - \Phi(\pg') = (1+1+\frac{1}{3}+\frac{1}{3}) - (\frac{1}{3}+\frac{1}{3}+1+0) = 1 $
\item[(ii)]$\Phi(\pg) - \Phi(\pg') = (1+1+\frac{1}{3}+0) - (\frac{1}{3}+\frac{1}{3}+0+\frac{1}{3}) = \frac{4}{3}$
\item[(iii)]$\Phi(\pg) - \Phi(\pg') = (1+1+\frac{1}{3}+\frac{1}{3}) - (\frac{1}{3}+\frac{1}{3}+0+0) = 2$
\end{itemize} }
\end{subfigure} 
\begin{subfigure}[c]{1.0\linewidth}
\centering
\includegraphics[width=0.8\linewidth]{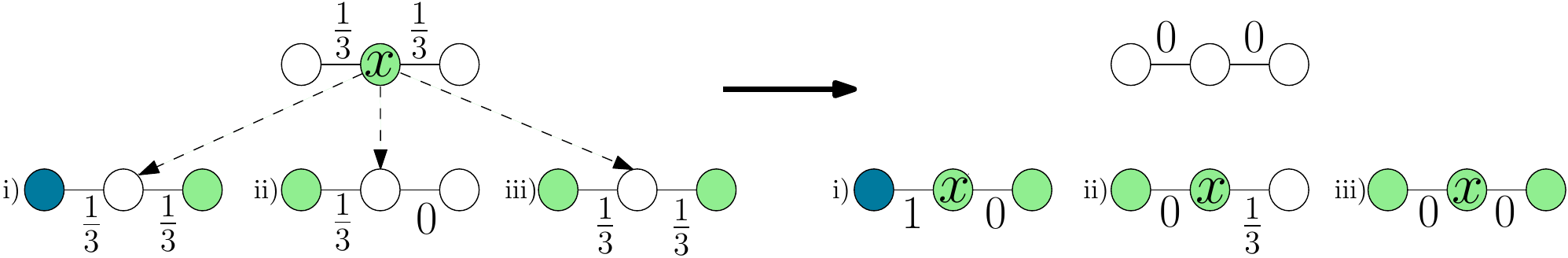}
\captionsetup{singlelinecheck=off}\subcaption{ \begin{itemize} \setlength{\itemindent}{5em} \vspace*{-0.5cm}
\item[(i)]$\Phi(\pg) - \Phi(\pg') = (\frac{1}{3}+\frac{1}{3}+\frac{1}{3}+\frac{1}{3}) - (0+0+1+0) = \frac{1}{3} $
\item[(ii)]$\Phi(\pg) - \Phi(\pg') = (\frac{1}{3}+\frac{1}{3}+\frac{1}{3}+0) - (0+0+0+\frac{1}{3}) = \frac{2}{3}$
\item[(iii)]$\Phi(\pg) - \Phi(\pg') = (\frac{1}{3}+\frac{1}{3}+\frac{1}{3}+\frac{1}{3}) - (0+0+0+0) = \frac{4}{3}$
\end{itemize} }
\end{subfigure} 
\begin{subfigure}[c]{1.0\linewidth}
\centering
\includegraphics[width=0.8\linewidth]{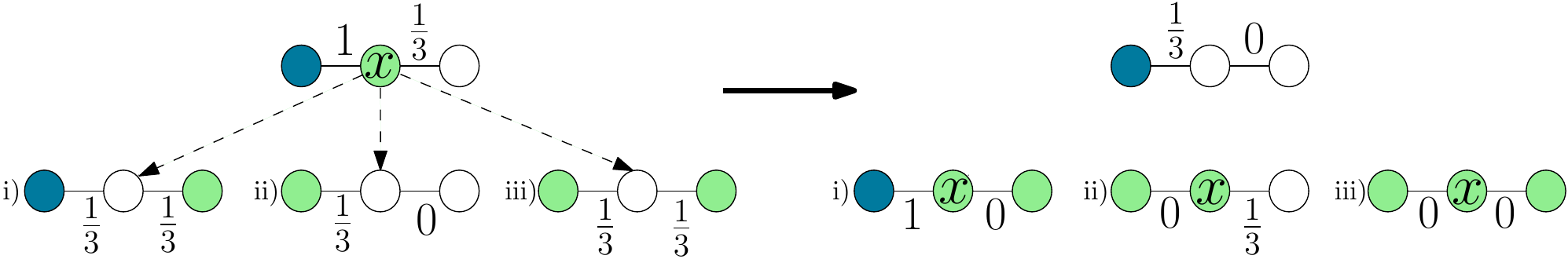}
\captionsetup{singlelinecheck=off}\subcaption{ \begin{itemize} \setlength{\itemindent}{5em} \vspace*{-0.5cm}
\item[(i)]$\Phi(\pg) - \Phi(\pg') = (1+\frac{1}{3}+\frac{1}{3}+\frac{1}{3}) - (\frac{1}{3}+0+1+0) = \frac{2}{3} $
\item[(ii)]$\Phi(\pg) - \Phi(\pg') = (1+\frac{1}{3}+\frac{1}{3}+0) - (\frac{1}{3}+0+0+\frac{1}{3}) = 1$
\item[(iii)]$\Phi(\pg) - \Phi(\pg') = (1+\frac{1}{3}+\frac{1}{3}+\frac{1}{3}) - (\frac{1}{3}+0+0+0) = \frac{5}{3}$
\end{itemize} }
\end{subfigure}
\begin{subfigure}[c]{1.0\linewidth}
\centering
\includegraphics[width=0.6\linewidth]{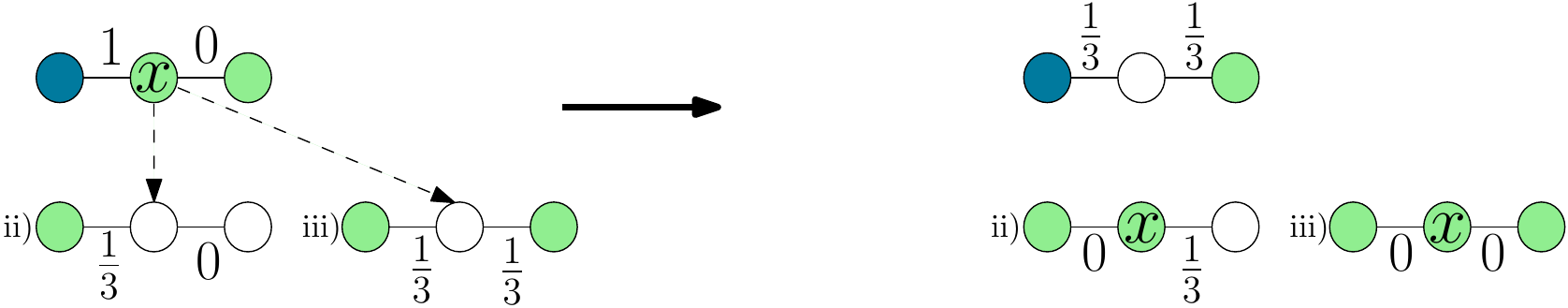}
\captionsetup{singlelinecheck=off}\subcaption{ \begin{itemize} \setlength{\itemindent}{5em} \vspace*{-0.5cm}
\item[(ii)]$\Phi(\pg) - \Phi(\pg') = (1+0+\frac{1}{3}+0) - (\frac{1}{3}+\frac{1}{3}+0+\frac{1}{3}) = \frac{1}{3}$
\item[(iii)]$\Phi(\pg) - \Phi(\pg') = (1+0+\frac{1}{3}+\frac{1}{3}) - (\frac{1}{3}+\frac{1}{3}+0+0) = 1$
\end{itemize} }
\end{subfigure}

\caption{Let $x \in A$ be an unhappy agent which makes an improving jump. We have the following possibilities for the new position of $x$: (i) One neighboring agent is of $x$'s type and the other is of the other type, (ii) one neighboring node is empty and the other one is occupied by an agent of $x$'s type and (iii) both neighboring agents of $x$ are of $x$'s type. $\Phi(\pg)=\sum_{e\in E} w_e$ decreases if $x$ makes an improving jump to some other node.}
\label{fig:ring_network}
\end{figure}
\end{proof}
 
\begin{theorem}\label{thm:IRC_cfJSG}
There cannot exist an ordinal potential function for the cf-JSG and the JSG for $\frac{1}{3} < \tau  \leq \frac{2}{3}$.
\end{theorem} 
\begin{proof}
We prove the statement by giving two examples of an improving response cycle for the cf-JSG on a grid. 
Since the cf-JSG is a special case of the JSG the statement holds for both variants.
\begin{itemize}
\item[(a)] $\mathbf{\tau\in (\frac{1}{3}, \frac{1}{2} ]}:$ Consider Fig.~\ref{fig:best_response_cycle_same} with $\tau\in (\frac{1}{3}, \frac{1}{2}]$.  The improving response cycle consists of six steps. We have green and blue agents and several (gray) empty nodes and the tolerance parameter $\tau$. All agents have the same favorite node shown in purple. In the first step the blue agent $a$ is unhappy and has cost of $(\tau-\frac{1}{3},3)$. By jumping next to node $b$ she becomes happy and reduces her cost to $(0,4)$. The green agent $y$ is happy and by jumping next to node $z$ and $b$ she gets closer to the common favorite node and remains happy as $\tau\in (\frac{1}{3}, 0.5]$. Thus $y$ reduces her cost to $(0,2)$. Now the green agent $z$ jumps to the common favorite node. The blue agent $a$ jumps back to her old node and reduces her cost from $(0,4)$ to $(0,3)$. Because of the movement of the blue agent the green agent $y$ becomes unhappy and jumps back to her old node and has cost $(0,3)$. In the next step the green agent $z$ jumps also back to her old node and reduces her cost from $(\tau,1)$ to $(0,2)$ which brings us in the same situation as in the beginning.
\begin{figure}[h!]
\centering
\includegraphics[width=0.80\textwidth]{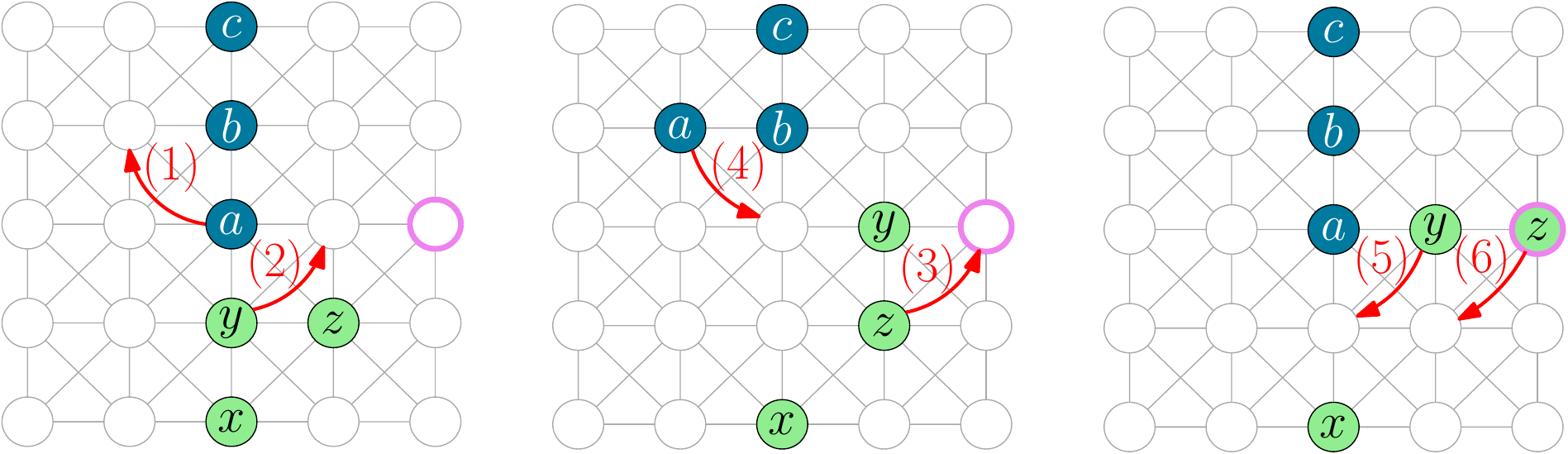}
\caption{\label{fig:best_response_cycle_same} An improving response cycle for the cf-JSG for $\tau\in (\frac{1}{3},\frac{1}{2}]$. Agents of type $A$ are blue, type $B$ agents are green. The common favorite node is purple.}
\end{figure}
\item[(b)] $\mathbf{\tau \in (\frac{1}{2},\frac{2}{3}]}:$ Consider Fig.~\ref{fig:best_response_cycle_same1} with $\tau \in (\frac{1}{2},\frac{2}{3}]$. The improving response cycle consists of six steps. We have green and blue agents and several (gray) empty nodes. All agents have the same favorite node shown in purple. In the first step the blue agent $a$ is unhappy and has cost of $(\tau-\frac{1}{2},3)$. By jumping next to node $b$ she becomes happy and reduces her cost to $(0,4)$. The green agent $y$ is happy and by jumping next to node $z$ she gets closer to the common favorite node and reduces her cost to $(0,2)$. Now the green agent $z$ jumps to the common favorite node. The blue agent $a$ jumps back to her old node and reduces her cost from $(0,4)$ to $(0,3)$. Because of the movement of the blue agent the green agent $y$ becomes unhappy and jumps back to her old node and has cost $(0,3)$. In the next step the green agent $z$ jumps also back to her old node and reduces her cost from $(\tau-\frac{1}{2}, 1)$ to $(0,2)$ which brings us in the same situation as in the beginning.
The improving response cycle for $\frac{1}{2}<\tau\leq  \frac{2}{3}$ can be found in Fig.~\ref{fig:best_response_cycle_same1}.
\begin{figure}[h!]
\centering
\includegraphics[width=0.95\textwidth]{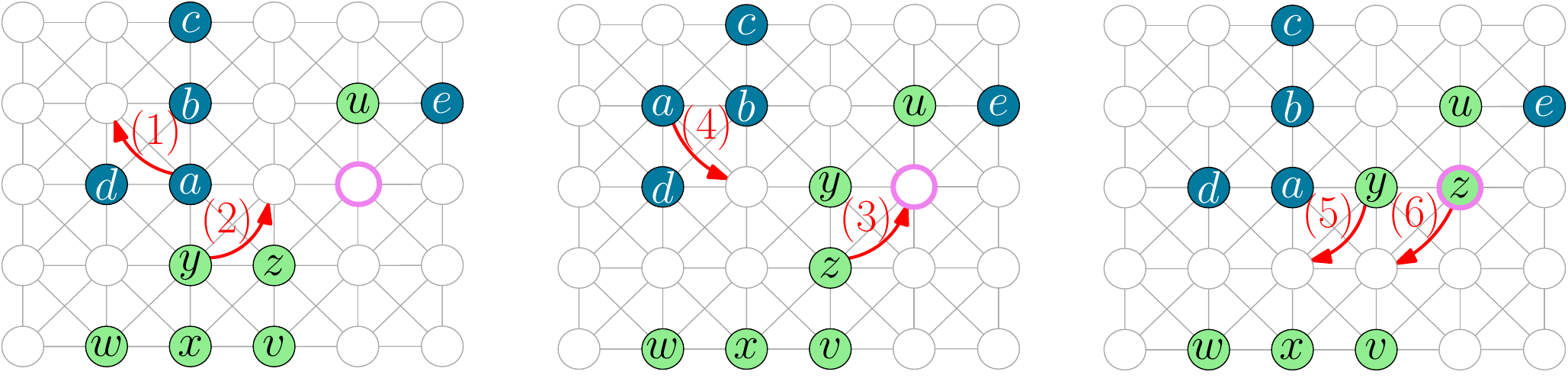}
\caption{\label{fig:best_response_cycle_same1} An improving response cycle for the cf-JSG for $\tau\in (\frac{1}{2},\frac{2}{3}]$. Agents of type $A$ are blue, type $B$ agents are green. The common favorite node is purple.}
\end{figure}
\end{itemize}  
\end{proof}

\section{Efficiency of Stable Placements}\label{sec:efficiency} 
In this section we investigate the properties of stable placements. In particular, we investigate their (in-)efficiency.

We start with proving that stable placements exist for many of our versions.
\begin{theorem}\label{thm:stable}
Stable placements exist for the u-SSG, the cf-SSG and the u-JSG.
\end{theorem}
\begin{proof}
We prove the existence separately.
\begin{itemize}
 \item[u-SSG:] Let $u \in V$ be a node, such that $\exists v \in V$ with $d_G(u,v) = D(G)$. Place all the agents of type $A$ sequentially in the following way: (1) Agent $a_1$ is placed at node $u$, (2) Place an agent $a_i$ at the node $u_i$ which has maximum value for $\frac{|N_w^+(\pg(a_i))|}{|N_w(u_i)|}$, 
(3) The agents of type $B$ are placed on the remaining nodes of $G$.  No happy agent will swap as it is not an improving move. Moreover there is no unhappy player of type $A$ who want to swap since the procedure ensures that all the agents of type $A$ have maximum possible local happiness ratio. Thus the placement is stable.
\item[cf-SSG:] The procedure of placement is analogue to the one we introduced for the u-SSG. The only difference is that we now consider the favorite node $u$ in $G$. Again we place all the agents of type $A$ sequentially in the following way: (1) Place agent $a_1$ at node $u$, (2) At step $i$, place an agent $a_i$ at the node $u_i$ which is closest to node $u$ and has maximum value for $\frac{|N_w^+(\pg(a_i))|}{|N_w(u_i)|}$. Then all the agents in $B$ are placed on the remaining nodes.  

It holds that every agent of type $A$ is either happy and closest to her favorite node or unhappy with maximum possible local happiness ratio. Thus no agent of type $A$ will swap. Since all agents have the same favorite node, no agents from $B$ will swap. Therefore we have a stable placement.
\item[u-JSG:] Let $u$ and $v$ be the nodes in the network $G$ such that $d(u,v)=D(G)$. Place all the agents of type $A$ sequentially in the following way: (1) Agent $a_1$ is at node $u$. (2) Place an agent $a_i$ on node $u_i$ which has maximum value for $\frac{|N_w^+(\pg(a_i))|}{|N_w(u_i)|}$ at iteration $i$.
Then all the agents in $B$ are sequentially placed in an analogous way starting from the node $v$ with using the local happiness ratio for each agent.  

All agents in $A$ either are happy or they are unhappy and have the maximum local happiness ratio. The same holds for agents in $B$. Thus, the placement is stable.
\end{itemize}
\end{proof}
Now we move on to proving basic properties of stable placements.
\begin{theorem}\label{thm:unhappy-property}
Let $\pg^{\NE}$ be a stable placement for the SSG on some graph $G$. The following statements hold:
\begin{itemize}
\item[(a)] If $D(G)>1$ and $\tau \leq\frac{1}{2}$ then for any $\pg^{\NE}$, at most one type of agents can be unhappy. Moreover, there exists a stable placement with unhappy agent(s).
\item[(b)] If $\tau> \frac{1}{2}$ there exist stable placements $\pg^{\NE}$ with unhappy agents of both types.
\item[(c)] For the $SSG$ there is a graph $G$ such that $\pg^{\NE}$ has a better total distance cost than the socially optimal placement. 
\item[(d)] For the $SSG$ with $\tau> \frac{1}{2}-\frac{1}{2n}$, for $n=|A|=|B|$, there exists a graph $G$ such that there is no placement where at least one agent is happy.
\end{itemize} 
\end{theorem}
\begin{proof}
\begin{itemize}
\item[(a)]We prove the first statement by contradiction. Assume that there exists a stable placement $\pg^{\NE}$ where $x\in A$ and $y\in B$ are unhappy. It holds that $$|N_w^+(\pg^{\NE}(x))| < \frac{1}{2} < |N_w^-(\pg^{\NE}(x))|$$ and $$|N_w^+(\pg^{\NE}(y))| < \frac{1}{2} < |N_w^-(\pg^{\NE}(y))|,$$ respectively. Thus, a swap between $x$ and $y$ will ensure that both agents are happy afterwards and therefore strictly decrease their cost. This contradicts the fact that $\pg^{\NE}$ is stable. 
See Fig.~\ref{fig:unhappy1} for an example of a stable placement with an unhappy agent. The placement is stable since all but one agents are content with $\pg$.
\begin{figure}[h!]
\centering
\begin{subfigure}[c]{0.48\textwidth}
\centering
\includegraphics[width=0.48\textwidth]{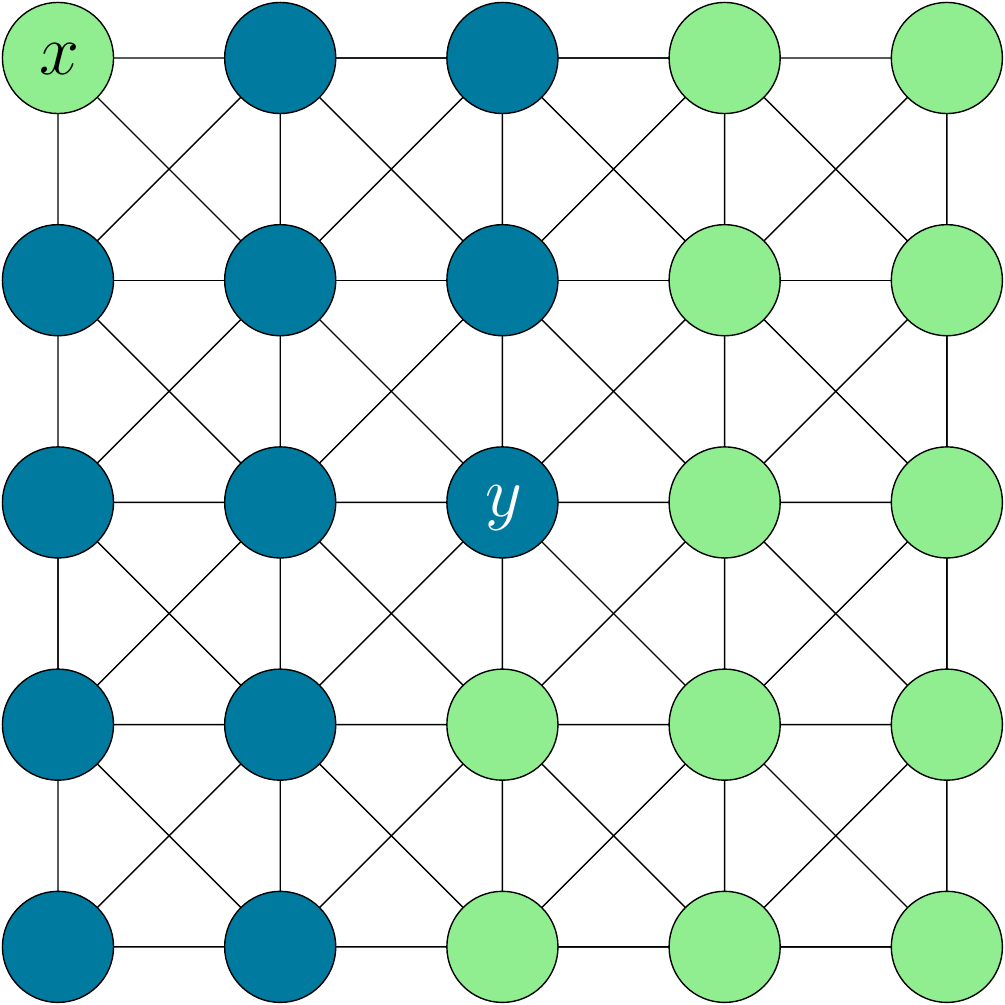}
\subcaption{$G$ with $\pg$, where every agent is at her favorite node. Agent $x$ is unhappy. }\label{fig:unhappy1} 
\end{subfigure}
~~~
\begin{subfigure}[c]{0.48\textwidth}
\centering
\includegraphics[width=0.48\textwidth]{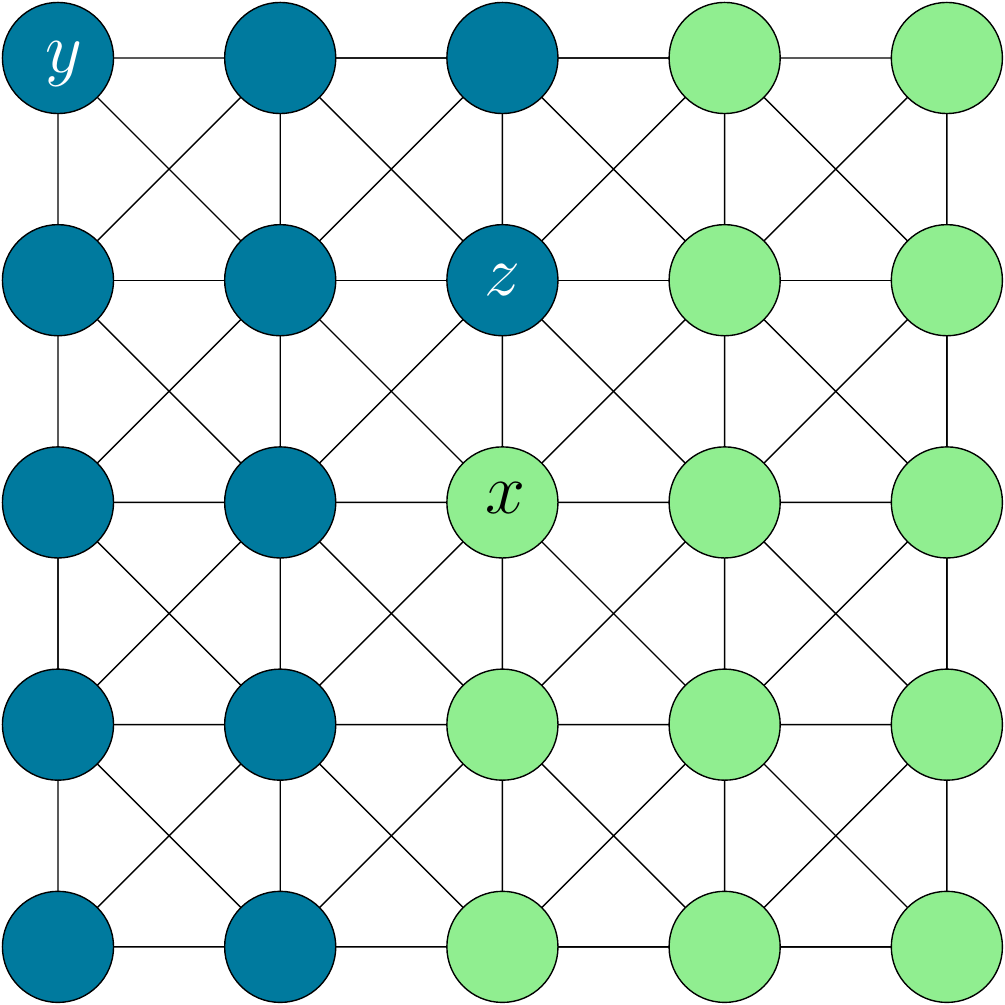}
\caption{Another placement for the instance shown in (a).
}\label{fig:happy}
\end{subfigure} 

 \begin{subfigure}[c]{0.30\textwidth}
\centering
\includegraphics[scale=0.35]{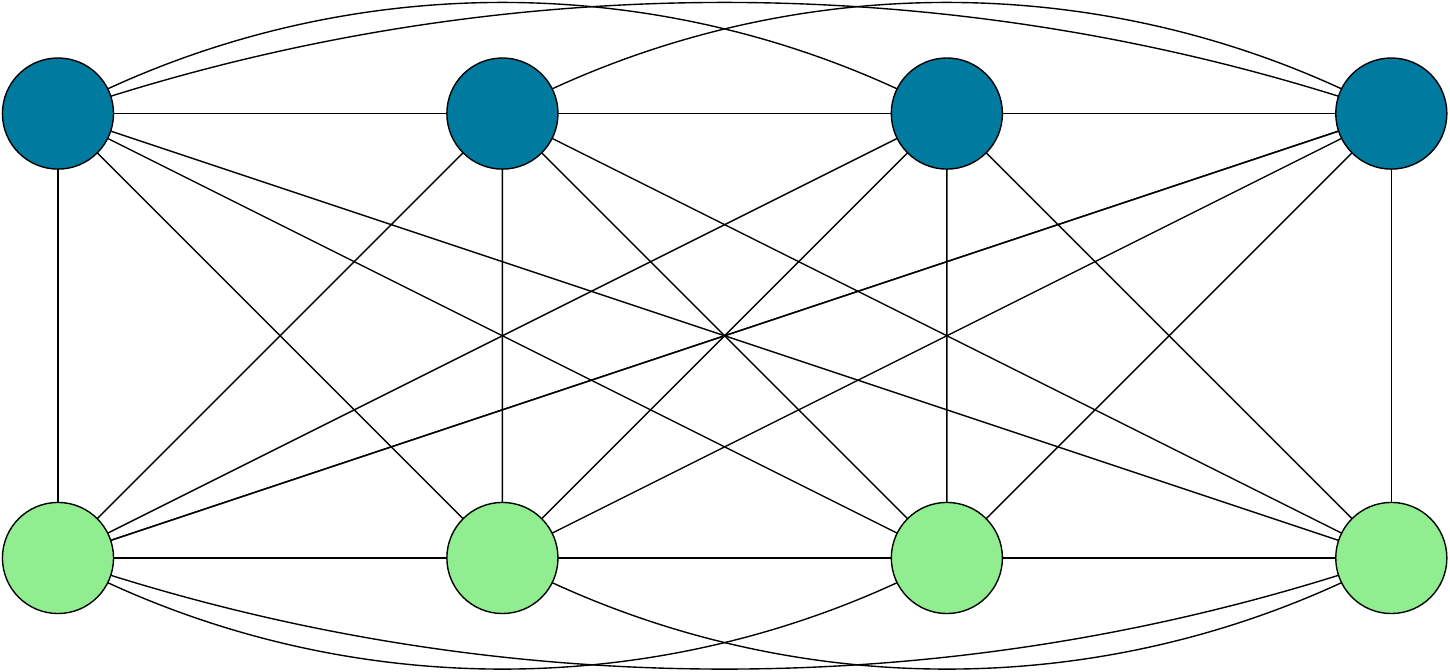}
\caption{$\tau >\frac{1}{2} - \frac{1}{8}$: In the $8$-clique $C$ in every placement every agent is unhappy. }\label{fig:clique}
\end{subfigure}

\caption{Examples of stable placements.} 
\end{figure}

\item[(b)] See Fig.~\ref{fig:happy} for an example of a stable placement. Agent $x$ and $z$ are unhappy. However no swap is possible.
\item[(c)] For $\tau \leq \frac{1}{2}$ all agents in the placement shown in Fig.~\ref{fig:happy} are happy. This placement is also a social optimum placement. Note, that the sum of distances from all agents to their favorite node is strictly smaller in the placement shown in Fig.~\ref{fig:unhappy1}.
\item[(d)]If we consider the $2n$-clique as the underlying network with $\tau>0.5-\frac{1}{2n}$ such that $|A|=|B|$, (cf. Figure \ref{fig:clique}). Independent of the  placement on $G$ every agent will be unhappy since always more than half of the neighborhood are of a different type.      
\end{itemize} 
\end{proof}
\noindent Remember that the social cost $\cost(\pg)$ of a placement $\pg$ is a vector. Hence we cannot use the state-of-the-art notions for the Price of Anarchy or the Price of Stability for investigating the efficiency of stable placements. For this, we first introduce suitable measures.
\begin{definition}
The \emph{ratio of happiness $\textit{RoH}$} and the \emph{ratio of distance $\textit{RoD}$} of two arbitrary placements $\pg$ and $\tilde\pg$ with $\textrm{Dist}(\pg'')= \sum_{u \in A\cup B} \dist_G(\fav_u,\pg''(u)) + 1$ is
\begin{align*}
 \RoH(\pg,\tilde\pg)  =\frac{\#\textrm{of unhappy players in } \pg+1}{\#\textrm{of unhappy players in } \tilde\pg+1} \textnormal{ and }
  \RoD(\pg,\tilde\pg)  =\frac{\textrm{Dist}(\pg)}{\textrm{Dist}(\tilde\pg)}.
\end{align*}
\end{definition}
Note that the $\RoH$ has the additional ``+1'' terms to handle the case where all agents in placement $\tilde{\pg}$ are happy. Essentially the $\RoH$ compares the social cost vectors of $\pg$ and $\tilde{\pg}$ by their first entry, and the $\RoD$ by their second entry. The $\RoD(\pg,\tilde\pg) = 1$ in the versions of u-SSG and u-JSG. 

We now define our notion of $\mathit{PoA}$ and $\mathit{PoS}$. 
\begin{definition}
For a given underlying network $G$, let $\NE$ be the set of all possible stable placements in $G$ and $\pg^{opt}$ be the socially optimal placement. Then,
\begin{align*}
 \mathit{PoA}&=\max_{\pg\in\NE}\{(\RoH\left(\pg,\pg^{\opt}\right), \RoD\left(\pg,\pg^{\opt}\right))\},\\
 \mathit{PoS}&=\min_{\pg\in\NE}\{(\RoH\left(\pg,\pg^{\opt}\right), \RoD\left(\pg,\pg^{\opt}\right))\}.
\end{align*}
\end{definition} 
\begin{observation}\label{obs:PoA}
For any $G$ we have $\mathit{PoA} \leq (n+1, D(G)-1)$ and $\mathit{PoS}\geq (1,1)$.     
\end{observation}

\begin{theorem}\label{thm:PoA-SSG-WD}
For the u-SSG with $\tau\leq \frac{1}{2}$, there exists a network $G$ and a stable placement $\pg^{\NE}$ such that in the worst case $\mathit{RoH}=\max\{|A|,|B|\}+1$ and in the best case $\mathit{RoH}=1$. 
\end{theorem}
\begin{proof}
To prove the upper bound we already know from Theorem \ref{thm:unhappy-property} that at most one type of agents can be unhappy. Thus the maximum number of unhappy agents in a placement is equal to the number of agents of the bigger group. When the optimal placement $\pg^\opt$ is stable the lower bound follows straightforwardly.

We show that the bounds are tight by an example.

Consider Fig.~\ref{fig:PoABadSSG-WD}. The green agents are happy since at least half of their neighbors are of their type. The placement is stable, since they have no incentive to swap. All blue agents are discontent. Notice that there are more blue agents than green ones. However, there is a stable placement, which is shown in Fig.~\ref{fig:PoA-OptSSG-WD} where all agents are content. The example can be easily extend to larger networks.
  
\begin{figure}[h] 
\centering
\begin{subfigure}[c]{0.49\textwidth}
\centering
\includegraphics[width=0.48\textwidth]{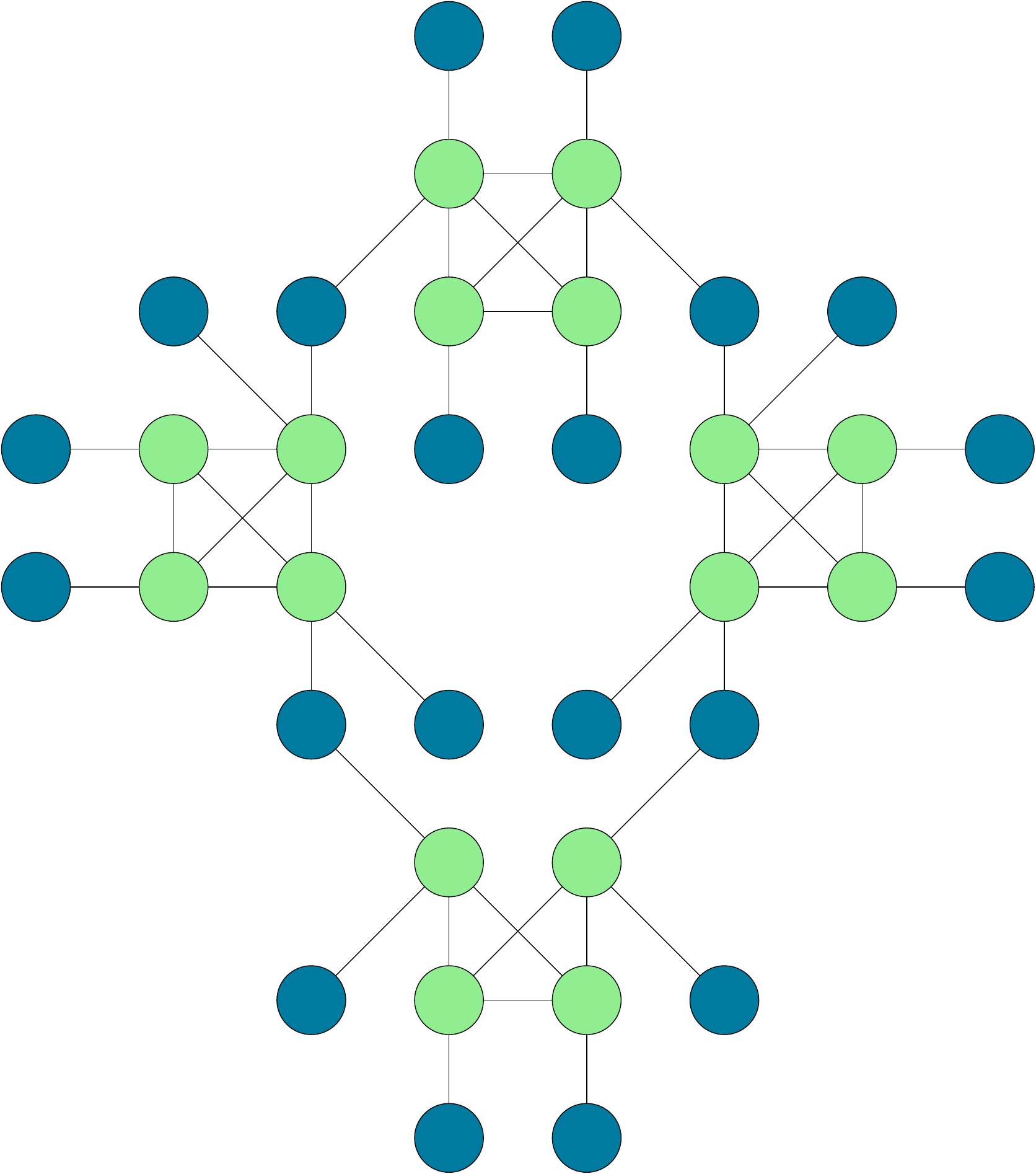}
\caption{$\mathit{RoH}=\max\{|A|,|B|\}+1$}\label{fig:PoABadSSG-WD}
\end{subfigure} 
\begin{subfigure}[c]{0.49\textwidth}
\centering
\includegraphics[width=0.48\textwidth]{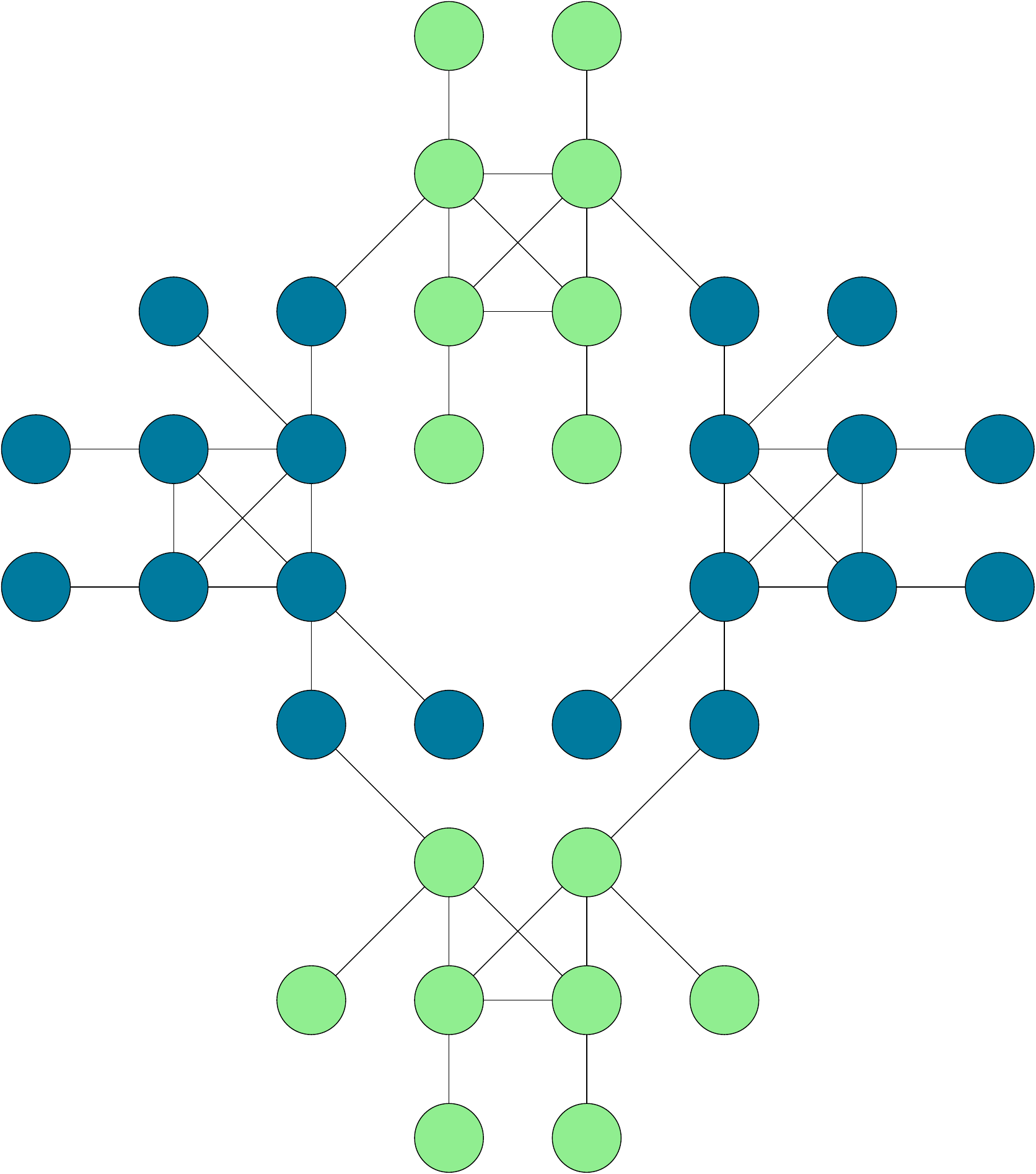}
\caption{$\mathit{RoH}=1$}\label{fig:PoA-OptSSG-WD}
\end{subfigure}
\caption{Two different stable placements in which in (a) all blue agents are unhappy whereas in (b) all agents are happy.}
\end{figure}   

\end{proof}
\noindent From Theorem \ref{thm:PoA-SSG-WD} the corollary below follows. 
\begin{corollary}
In the u-SSG there exists a network $G$ such that the $PoA=(max\{|A|,|B|\}+1,1)$ and the $PoS=(1,1)$.
\end{corollary}

Now we give a tight bound for the Price of Anarchy.
\begin{theorem}\label{thm:tight_PoA}
The $SSG$ has $PoA\in (\Theta(n),\Theta(D(G)))$ for some network $G$.
\end{theorem}
\begin{proof}
Consider the network $G$ and the placement in Fig.~\ref{fig:PoABad-tight}(top). We assume that $w \in \Theta(1)$ and the number of agents of type $A$ is $$|A| = |B|+w+(w-1)\left\lceil\frac{|B|}{2}\right\rceil+w\left\lfloor\frac{|B|}{2}\right\rfloor.$$ Let $n = |A| + |B|$ be the total number of agents.
The favorite node of agent $a_i$ is $\pg(b_i)$ and of agent $b_i$ $\pg(a_i)$, respectively, with $i = 1, \ldots, |B|$. All other agents are placed at their favorite node.

The placement is stable. Agents of the same type will not swap since this would be a setback for at least one of the involved agents. Moreover there is no swap possible between $a_i$ and $b_i$ since $b_i$ would just have agents of type $A$ in her neighborhood. For $\frac{1}{2w}<\tau< 1-\frac{1}{2w}$ all agents of type $A$ are happy and all agents of type $B$ are unhappy.

The distance between $a_i$ and $b_i$ is at least $(|B|+w)\in \Omega
(D(G))$.  Therefore, the $cost(\pg)$ is 
\[cost(\pg)\in \left(|B|, |B|\cdot\Omega(D(G))\right) =\left(\Omega(n), \Omega(n \cdot D(G))\right).\]
However in the social optimum placement (cf. Fig.~\ref{fig:PoABad-tight}(bottom)) it can be observed that every agent is happy and is at her favorite location. Thus the cost of optimum is $(0,n)$.
Now we are ready to give bounds on $\RoH$ and $\RoD$. Observation~\ref{obs:PoA} together with the cost of the stable and optimum placement yields $\RoH \in\Theta(n)$ and $\RoD \in \Theta(D(G))$. Thus, the Price of Anarchy is $PoA\in(\Theta(n),\Theta(D(G)))$. 
\end{proof}
\begin{figure}[h!]
\centering
\includegraphics[width=0.9\textwidth]{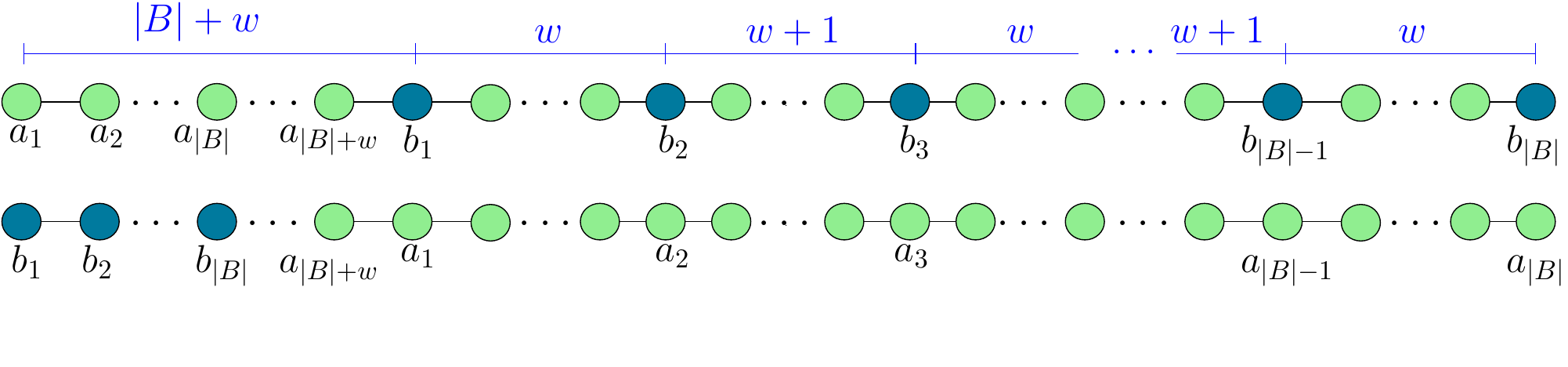}\vspace*{-0.7cm}
\caption{A stable placement (top) and the social optimum placement (bottom).}\label{fig:PoABad-tight}
\end{figure}
\begin{remark}
 The construction in the proof of Theorem~\ref{thm:tight_PoA} can be easily extended to the JSG. Thus, also for the JSG $PoA\in(\Theta(n),\Theta(D(G)))$ for some network~$G$. 
\end{remark}
\section{Conclusion and Open Questions}
In this work we have introduced the first truly game-theoretic version of the well-known Schelling segregation model. The selfish agents in our model strategically choose their locations and take the structure of their local neighborhood as well as  their individual location preferences into account. 

We have established that many variants of our model actually are potential games, which implies the existence of pure (2-coalitional) Nash equilibria and guaranteed convergence in the sequential version. However, we have also identified cases, e.g. the (cf)-JSG, which have improving response cycles. Moreover, we have investigated the efficiency of stable placements in the (u-)SSG and proved high tight bounds on the Price of Anarchy. This implies that the outcomes of our game vary significantly in their social cost. 

The most intriguing open problem is to settle the convergence behavior for the SSG for $\tau < \frac{1}{2}$ and the u-JSG. We conjecture, that a potential function exists. Moreover, it is open whether stable placements exist for the SSG and the (cf-)JSG. We conjecture that a stable placement exists for all variants. Another ambitious endeavor would be to prove bounds on the size of the monochromatic regions similar to the works~\cite{BIK12,BEL14,BEL16,BIK17}. In particular, it would be interesting to explore the impact of the location preferences on the induced stable placements. 

\bibliographystyle{abbrv}
\bibliography{Schelling}

\end{document}